\documentclass[conference]{IEEEtran}
\IEEEoverridecommandlockouts

\usepackage{amsmath,amsfonts}
\usepackage{algorithm}
\usepackage{algorithmicx}
\usepackage{array}
\usepackage[caption=false,font=normalsize,labelfont=sf,textfont=sf]{subfig}
\usepackage{textcomp}
\usepackage{stfloats}
\usepackage{url}
\usepackage{verbatim}
\usepackage{graphicx}
\def\BibTeX{{\rm B\kern-.05em{\sc i\kern-.025em b}\kern-.08em
    T\kern-.1667em\lower.7ex\hbox{E}\kern-.125emX}}

\usepackage[utf8]{inputenc}
\usepackage{amsmath}
\usepackage{amsfonts}
\usepackage{amssymb}
\usepackage{amsthm}
\usepackage{color}
\usepackage{algorithm}
\usepackage{algpseudocode}
\usepackage{bbm}
\usepackage{bm}
\usepackage{graphicx}
\usepackage{lipsum}
\usepackage{bbold}
\usepackage{cite}
\usepackage[hidelinks]{hyperref}

\newcommand{\Hhat}{\hat{H}}
\newcommand{\Pmax}{|P|_{max}}

\newcommand{\phat}{\hat{p}}
\newcommand{\prob}{\mathbf{P}}
\newcommand{\E}{\mathbb{E}}

\newcommand{\Mbar}{\bar{M}}

\newtheorem{Lemma}{Lemma}  
\newtheorem{Theorem}{Theorem}
\newtheorem{Assumption}{Assumption}

\newtheorem{Example}{Example}

\begin{document}

\title{Learning the Influence Graph of a Markov Process that Randomly Resets to the Past\\
}
\author{\IEEEauthorblockN{Sudharsan Senthil, Avhishek Chatterjee}
\IEEEauthorblockA{Indian Institute of Technology Madras, India\\
sudharsansenthil@hotmail.com, avhishek@ee.iitm.ac.in}
}
\maketitle

\begin{abstract}
Learning the influence graph $G$ of a high-dimensional Markov process is central to many application domains, including social networks, neuroscience, and financial risk analysis. However, in many of these applications, future states of the process are occasionally and unpredictably influenced by a distant past state, thus destroying the Markovianity. To study this practical issue, we propose the past influence model (PIM), which captures the occasional "random resets to past" by modifying the Markovian dynamics in \cite{bagewadi2024learninginfluencegraphhighdimensional},  which, in turn, is a non-linear generalization of the dynamics studied in \cite{8345299}, \cite{6736512}. The recursive greedy algorithm proposed in this paper recovers any bounded degree $G$ when the number of ``jumps back in time" is order-wise smaller than the total number of samples, and the algorithm does not require memory.
\end{abstract}

\begin{IEEEkeywords}
Non-Markov Learning and Graph Learning.
\end{IEEEkeywords}

\section{Introduction}
Learning the underlying influence (causal) graph $G$ of a high-dimensional multivariate stochastic process has applications in multiple domains, including online social networks, neural and biological signal processing, and financial risk analysis \cite{quinn2011estimating,degroot1974reaching,ravazzi2021ergodic,hall2019learning}. The existing literature focuses on learning the influence graphs of Markov stochastic processes. Motivated by important phenomena in social networks and financial risk analysis, this paper studies the problem of learning the influence graphs of a class of stochastic processes whose distant past unpredictably influences their future.

In a multivariate Markov dynamic with an influence graph, the next state of a process variable depends on the current states of process variables that have directed edges to it. We study the following modification of such a Markov process. The future states depend on the current states with probability $p$, and with probability $1-p$ they depend on the states $d$ time-steps back. Intuitively, it is a process in which a Markovian evolution is randomly reset to a distant past state. The task is to recover the influence graph from samples of the process observed over $T$ time-steps, without knowing the instances when the “random resets to the past” occurred.

This phenomenon is observed in many domains, where the distant past occasionally influences the future, in addition to the present. For instance, in social networks, users may sometimes form opinions based on older interactions rather than the latest activities. In many social networks, agents express their opinions in a binary form (likes/dislikes), which in turn influences the beliefs of others within the network. Occasionally, opinions expressed in the distant past resurface and shape the opinions of others. Similarly, in financial markets, sudden shocks or crises often resurface after long intervals, influencing present-day risk perception and investment behavior. 

The social network model proposed in this paper addresses the situation in which the entire network is influenced by its past, rather than just a subset of participants. We also model the reactions of the network participants with a binary alphabet, though the framework can be easily extended to problems involving larger observable symbol sets. 

Though analytical guarantees for learning processes with random resets to the past have not been studied before to our best knowledge, a large body of work has addressed the problem of learning i.i.d. graphical model structures. Early work by Chow and Liu \cite{chow1968dependence} introduced tree-structured approximations for discrete distributions, and subsequent efforts extended these ideas to bounded tree-width Markov networks \cite{karger2001markov}. In high-dimensional settings, algorithms based on sparsity have been popular, including $\ell_{1}$-regularized logistic regression for Ising models \cite{ravikumar2010ising} and group-sparse regularization techniques \cite{jalali2011group}. Other approaches exploit structural properties: Bresler et al. \cite{bresler2008markov} proposed reconstruction methods for Markov random fields, while Wu et al. \cite{wu2013loosely} investigated loosely connected graphs. Greedy methods have also proven effective for scalable structure recovery \cite{netrapalli2010greedy}. And, more recently \cite{DaskalakisKandirosYao2025} provided algorithms for learning directed Gaussian graphical models with sample complexity. 

The algorithm proposed in this paper is for learning processes with random resets to the past, which is a generalization of the \textbf{RecGreedy} algorithm introduced in \cite{7097023} for i.i.d. samples. Recently, a similar approach was adopted in \cite{bagewadi2024learninginfluencegraphhighdimensional} to learn Markovian processes without random resets to the past, but with a finite memory.

The paper is organized as follows: In Section~\ref{sec: model} we introduce the \textbf{Past Influence Model}, and describe how it evolves over discrete time steps. In Section~\ref{sec: ps} we define the learning problem, followed by Section~\ref{sec: alg}, where we propose the \textbf{PIMRecGreedy} algorithm, explaining how it works and why it is more effective than simply considering a Markov process with a large enough memory, along with its sample complexity theorem. Section~\ref{sec: sims} presents the performance of \textbf{PIMRecGreedy} across varying sample sizes. 

\section{The Social Network Model} \label{sec: model}
\subsection{Modeling}

Influence is modeled by a directed graph $G = (V, E)$, where nodes represent agents and directed edges represent influence. A directed edge $(u, v) \in E$ means $u$ influences $v$, and self-influence is modeled by $(v, v) \in E$. The neighborhood of $v$ is $\mathcal{N}_v := \{u \neq v : (u, v) \in E\}$.

Time evolves in discrete steps. At each \(t\), node \(v\) generates \(M_v(t) \sim \min(\text{Poisson}(\mu_v(X_v(t))), \overline{M})+1\) independent Bernoulli samples with success probability \(X_v(t)\), where \(\mu_v:[0,1]\to\mathbb{R}\) is \(L\)-Lipschitz and \(\overline{M}\) is a constant cap. Out of these samples, \(N_v(t)\) are ones.  

Each directed edge \((u,v)\) has weight \(a_{uv}\in(0,1)\), with \(\sum_u a_{uv}=1\). The internal parameter \(X_v(t)\) evolves via Eqn.~(\ref{eq: node_param}) through three components: neighborhood observations \(\mathcal{N}_v\), intrinsic bias \(l_v\), and random fluctuation \(Z_v(t)\in[0,1]\) with mean \(0<\bar{z}_v<1\) and parameter \(\beta\in(0,1)\). \(C(t)\sim \mathrm{Ber}(p)\) models the random reset to the past $(t-d)$. The coefficient \(\alpha_v\) measures \(v\)’s openness to external influence, trading off social input against internal dynamics. We call this the \textbf{PIM (Past Influence Model)}. We sometimes refer to the time-reset as occurrence of Tail (in a coin toss perspective).

\begin{align}
\label{eq: node_param}
\nonumber
    &X_v(t+1) = (1-\alpha_v)\left[ (1 - \beta)Z_v(t) + \beta l_v\right] \\ & + \alpha_v \sum_{u \in \mathcal{N_v} \cup v} \alpha_{uv}\left[C(t)\frac{N_u(t)}{M_u(t)} + (1-C(t))\frac{N_u(t-d)}{M_u(t-d)} \right]
\end{align}

\subsection{Some Mathematical Structures}

Given the observations $\{(N_v(t), M_v(t)) : v \in V\}_{t=0}^{T-1}$, we define the variable $Y_v(t) := \frac{N_v(t)}{M_v(t)}$. We also define:

\begin{itemize}
    \item $Y_Q(t) := \{Y_v(t) : v \in Q \subseteq V\}$
    \item $Y_{u,Q}(t) := \{Y_v(t) : v \in Q \cup \{u\}, Q \subseteq V\}$
    \item $Y(t) := \{Y_v(t) : v \in V\}$
\end{itemize}
\vspace{3mm}
If there was no coin toss in Eqn.~(\ref{eq: node_param}) the process $Y(t)$ becomes a normal Markov chain, $\{W(t)\}$. The directed conditional entropy of node $v$ given a set $Q$ is defined as:
\[
-\sum_{Y_{Q,v}(t)} \sum_{Y_v(t+1)} \mathbf{P}(Y_v(t+1),Y_{Q,v}(t)) \log \mathbf{P}(Y_v(t+1)|Y_{Q,v}(t)),
\]
for some $Q \subset V$ and $v \not\in Q$. In a stationary setting, this quantity becomes time-invariant and is denoted as $H(v_+|v,Q)$. This can be computed by first calculating the joint entropies $H(v_+,v,Q)=
-\sum \mathbf{P}(Y_v(t+1),Y_{Q,v}(t)) \log \mathbf{P}(Y_v(t+1),Y_{Q,v}(t))$ and $H(v,Q)$ and $H(v,Q) = -\sum_{Y_{Q,v}(t)}  \mathbf{P}(Y_{Q,v}(t)) \log \mathbf{P}(Y_{Q,v}(t))$, in the former, the summation is over $Y_{Q,v}(t)$ and $Y_v(t+1)$.

In practice, the conditional entropy \(H(v_+|v,Q)\) cannot be computed directly due to unknown distributions. Instead, we estimate it from data and denote the estimate by \(\hat{H}(v_+|v,Q)\). Specifically, we compute plug-in estimates of the joint entropies \(\hat{H}(v_+,v,Q)\) and \(\hat{H}(v,Q)\) using the empirical distributions \(\hat{p}(y_v,y_Q)\) and \(\hat{p}(y_{v_+},y_v,y_Q)\), and then take their difference.  

\subsection{Probabilistically teleporting to the past}
\label{sec:PIM_intro}
 At each discrete time step \(t\), the PIM, according to Eqn.~(\ref{eq: node_param}), performs a Bernoulli trial with probability of success \(p\). If the outcome is a tail, the model \emph{ignores} the most recent \(d\) samples and instead evolves based on the samples from time \(t-d\). Hence, the model effectively follows the distribution from time \(t-d+1\), which we refer to as a \emph{reset to the past}. In this way, the true state of \(Y(t+1)\) aligns with the true state at time \(t-d+1\).  

\begin{align}
\label{sq: obs}
    &y(1), y(2), .., y(d+1), 
    \underset{\textbf{\textcolor{red}{Tail}}}{y^{(2)}(2)}, 
    \underset{\textbf{\textcolor{red}{Tail}}}{y^{(2)}(3)}, 
    .., y^{(2)}(d+1), \\ \nonumber
    &y(d+2), 
    \underset{\textbf{\textcolor{red}{Tail}}}{y^{(3)}(3)}, 
    .., 
    \underset{\textbf{\textcolor{red}{Tail}}}{y^{(3)}(d+2)}, 
    y(d+3), 
    \underset{\textbf{\textcolor{red}{Tail}}}{y^{(4)}(4)}, ..
\end{align}

In the above run of \(Y(t)\), the model forgets the recent past at the \((d+2)^{\text{th}}\) time step. As a result, the sample \(Y^{(2)}(2)\) is drawn from the same distribution as \(Y(2)\); similarly, \(Y^{(3)}(3)\) is distributed identically to both \(Y^{(2)}(3)\) and \(Y(3)\).  

The observations (\ref{sq: obs}) are then processed as in~(\ref{eq:PIM_pairs}) to estimate probabilities. Since the PIM does not know the positions of the tails, mismatched samples (e.g., \(y^{(2)}(2)\) generated from \(y(1)\)) are also erroneously included in the estimation. Nevertheless, with sufficiently large sample size (cf. Theorem~\ref{thm:2}), the estimation error remains small enough to reliably separate neighbors from non-neighbors.

\begin{align}
\label{eq:PIM_pairs}
   &[y(1),y(2)], .., 
   \underset{\textbf{\textcolor{red}{Tail}}}{[y(d+1),y^{(2)}(2)]}, 
   \underset{\textbf{\textcolor{red}{Tail}}}{[y^{(2)}(2),y^{(2)}(3)]}, .. \\ \nonumber
   &[y^{(2)}(d+1),y(d+2)], 
   \underset{\textbf{\textcolor{red}{Tail}}}{[y(d+2),y^{(3)}(3)]}, 
   [y^{(3)}(3),y^{(3)}(4)], .. \\ \nonumber
   &\underset{\textbf{\textcolor{red}{Tail}}}{[y^{(2)}(d+1),y^{(3)}(d+2)]}, .., 
   \underset{\textbf{\textcolor{red}{Tail}}}{[y(d+3),y^{(4)}(4)]}, ..
\end{align}

\section{The Learning Problem} \label{sec: ps}
The central problem in this paper is to learn the underlying influence graph based on the observed variables over a finite time window $[0, T-1]$. That is, given the samples$\{(N_v(t),{M}_v(t)): v \in V\}_{t=0}^{T-1}$  learn the neighborhood $\mathcal{N}_v$ for all nodes $v \in V$.

\section{Algorithm to Recover the Influence Graph} \label{sec: alg}
In this section, we describe a greedy approach to recover the underlying directed graph $G$ using the observation sequence $\{(N_v(t), M_v(t)): v \in V\}_{t=0}^{T-1}$. This procedure is a variant of the RecGreedy algorithm proposed for learning graphical models with i.i.d.\ samples~\cite{7097023}. Our primary technical contribution lies in establishing a sample complexity guarantee for Markovian dynamics that probabilistically gets reset to the past. Specifically, we demonstrate that the algorithm can recover the true graph with high probability using data collected over $\mathcal{O}(d^2 + d\log |V|)$ time steps.

Algorithm~1 is similar in spirit to the algorithm proposed in~\cite{bagewadi2024learninginfluencegraphhighdimensional}, but that relies on estimating entropies of the form $\Hhat(Y_{v+} \mid Y^m_{v,Q})$, where 
\[
Y^m_{v,Q}(t) = \{Y_{v,Q}(i): i \in [t-m,t]\},
\]
where $m$ is the Markov chain's memory. For higher values of $d$, following the approach of~\cite{bagewadi2024learninginfluencegraphhighdimensional} with $m=d$ becomes computationally very challenging, as probability estimation is required over large conditioning sets. By contrast, the PIMRecGreedy algorithm presented here requires no memory. When the sample size $T$ satisfies Theorem~\ref{thm:2}, the algorithm recovers the graph with a very high probability.

The algorithm relies on estimates of directed conditional entropy, $\Hhat(v_+ \mid v,Q)$, for various $Q \subset V$. For each node $v$, it begins with empty sets $\hat{U}(v)$ and $\hat{T}(v)$, representing intermediate and final neighborhood estimates. At each iteration, the algorithm selects the node $u$ that maximizes the reduction in directed conditional entropy, i.e.,
\[
\Hhat(v_+ \mid v,\hat{U}(v)) - \Hhat(v_+ \mid v,\hat{U}(v),u),
\]
and includes $u$ in $\hat{U}(v)$ if the reduction exceeds a fixed threshold. This process continues until no further node yields a reduction exceeding the threshold.

When the threshold is set to $\kappa/2$, Lemma~\ref{lem:tau_modified} guarantees that the last node added to $\hat{U}(v)$ is always a true neighbor of $v$. Only this last added node is included in the final estimate $\hat{T}(v)$. The procedure is then repeated, initializing with $\hat{U}(v)=\hat{T}(v)$. The process terminates once $\hat{T}(v)$ matches the true neighborhood, as no additional nodes satisfy the threshold condition.

The algorithm's sample complexity is largely determined by the dimension of the probability distributions that need to be estimated from data. Since it proceeds from smaller subsets P and stops when no further decrease in entropy is observed, the dimension of the involved distributions remains bounded.

To analyze the sample complexity, we use (Lemma~5,~\cite{bagewadi2024learninginfluencegraphhighdimensional}),
\[
\Pmax \leq \frac{2 \log \!\left( \tfrac{ \bar{M} \left( \bar{M} + 1 \right) }{2} + 2 \right) }{ \kappa } + 1
\]
to denote the largest cardinality of subset $P$ required \footnote{The bound is obtained by exploiting the fact that the inner loop of the algorithm reduces the entropy by at least $\kappa/2$ upon adding a new node to $\hat{U}(v)$.}. Let $\chi$ denote the support of the random variable $Y_v(t)$, and define $|\xi| := |\chi|^{(2+ \Pmax)}$ to be the maximum support size of the joint distributions to be estimated. Together with properties of the Markov process, this yields a logarithmic sample complexity, formally established in Theorem~1.

Let $p_g$ denote the probability distribution of the observed samples given the instants of tails and the observations are ordered (in time) accordingly, we shall call this the post-processing distribution. This distribution is used only for the sample complexity analysis. For $M_v(t)=1$. it can be shown that $\sum_{w_{v_+}} |p_g(y_{v_+}|y_{v,Q}) - p_g(y_{v_+}|y_{v,Q,u})|$ is bounded away from zero, for  $ \forall v \in V$, any values of $y_v,y_Q$, when $u \in \mathcal{N}_v$, and $Q \subset V$ such that $u \not\in Q$. For the general case, due to certain technicalities, we make the following assumption.

\begin{Assumption}
Consider any $v \in V$, $u \in \mathcal{N}_v$, and $Q \subset V$ such that $u \not\in Q$. Then for any values $y_v$ and $y_Q$, there exists a value $y_u$ such that for some constant $\kappa' > 0$,
\[
\sum_{y_{v_+}} |p_g(y_{v_+}|y_{v,Q}) - p_g(y_{v_+}|y_{v,Q,u})| >  \sqrt{\kappa'}
\]
\end{Assumption}   
\begin{Theorem}
\label{thm:2}
Let $\tilde{A}$ be the matrix whose $ij^{th}$ element is $\alpha_i \alpha_{ij}$, and $\rho(\tilde{A})$ is the largest eigenvalue of $\tilde{A}$. Also, let $\bar{\mu} := \max_{v} \sup_{x \in [0,1]} \mu_v(x).$ The PIMRecGreedy($\kappa$) algorithm recovers the true influence graph with probability at least $1-\gamma$ if the sample size satisfies the following constraint:
\begin{align*}
  &  T  \geq\max \Biggl\{ 
    \log \!\left( \frac{c + 2 \cdot |V|^{(P_{\max} + 1)} \, |\xi|}{\gamma} \right) 
    \frac{12 \delta'\beta_1}{c_1^2}, \;  \\
  & (d^2 - d + 2) + \frac{(1 - 2(\bar{\mu} + L) |\rho(\tilde{A})|)\delta^2}{2 (1 + 2(\bar{\mu} + L) |\rho(\tilde{A})|) |\xi|^2(d+1)^{-1}} 
  \\[6pt] 
  & \quad \times \log \!\left( \frac{c + 2 \cdot |V|^{(P_{\max} + 1)} \, |\xi|}{\gamma} \right)
    \Biggr\}
\end{align*}

  \text{Where, }$c>0, |\xi|:=|\chi|^{(2+\Pmax)}$, $|\chi| \leq (\frac{\Mbar(\Mbar+1)}{2}+2)$, $\beta_1\delta' - 4(T-1)^{\alpha - 1} > c_1$ ,$\delta. \log\left( \frac{|\xi|}{\delta}\right) \leq \frac{\epsilon}{4}$, $\delta'. \log\frac{|\xi|}{\delta'} \leq \frac{\epsilon'}{4}$ and $(1-p) = \frac{(T-1)^{\alpha}}{\beta_1(T-d-1)}$, for some $\alpha < 1, \text{ and } 0<\beta_1 <1$, $2(\bar{\mu} + L) |\rho(\tilde{A})|< 1$, and Assumption 1 holds. 

\end{Theorem}

\begin{algorithm}
\caption{PIMRecGreedy($\kappa$) Algorithm}\label{alg_recgreedy}
 \hspace*{\algorithmicindent} \textbf{Input}: $\{N_v(t) : v \in V\}_{t=0}^{T-1}, \{M_v(t) : v \in V\}_{t=0}^{T-1}, $ \\
 \hspace*{\algorithmicindent} \textbf{Output}: ${\hat{T}(v):  \forall v \in V}$
\begin{algorithmic}[1]
 \For{$v \in V$}
    \State $\hat{T}(v), \leftarrow \phi$
    \Repeat
        \State $\hat{U}(v) \leftarrow \hat{T}(v)$
        \State $\text{LastNode} \leftarrow \phi$
        
            \Repeat
            \State{\vspace{-0.7cm}
            \begin{align*}
            u &= \arg \max \limits_{k \in V \setminus \hat{U}(v)} \hat{H}(v_+|v,\hat{U}(v)) \\ &- \hat{H}(v_+|v,\hat{U}(v),k) 
            \end{align*}}
            \State {$\delta'_u =  \hat{H}(v_+|v,\hat{U}(v)) - \hat{H}(v_+|v,\hat{U}(v),u) $}
              \If{$\delta'_u > \kappa/2$}             
            
                    \State $\hat{U}(v) \leftarrow \hat{U}(v) \cup \{u\}$
                    \State $\text{LastNode} \leftarrow u$
                
              \Else
                \State $\hat{T}(v) \leftarrow \hat{T}(v) \cup \text{LastNode}$
              \EndIf 
            
        \Until New node is not added to $\hat{U}(v)$
    \Until New node is not added to $\hat{T}(v)$
\EndFor
            
\end{algorithmic}
\end{algorithm}

\section{Simulation results} \label{sec: sims}
The data are generated using Eqn.~(\ref{eq: node_param}) with parameters 
$l_v=0.167$, $\beta=0.75$, $\alpha_v=0.8$, $\beta_1=0.75$, $\alpha=0.5$, 
$|V|=5 \text{ and } 10$, $\mu_v(x)=0.4x$, and $\alpha_{uv}=1/|\mathcal{N}_v|$ for 
$u \in \mathcal{N}_v$. Performance of \textbf{PIMRecGreedy} over ring 
and line graphs are plotted in Fig.~\ref{fig:1} for $d=5$ and 
$\overline{M}=1$. The simulations are repeated for a slightly harder case 
with $d=10$ and $\overline{M}=2$, and the results are shown in 
Fig.~\ref{fig:2}. The simulation results for the graphs presented in Fig. 3 are shown in Fig. 4.

\subsection*{Choice of Parameter $\kappa$}\label{sec: cross}
The threshold parameter $\kappa$ for \textbf{PIMRecGreedy} is selected via cross-validation on a labeled/held-out dataset, by choosing the value that maximizes a suitable performance metric (here, the probability of full graph recovery). Similar parameter tuning is required in other graphical model learning methods, such as the regularization parameter $\lambda$ in $RWL$ \cite{ravikumar2010ising} and the local separator size $\eta$ in the $CVDT$ algorithm \cite{anandkumar2012ising}. Cross-validation results on ring and line graphs with 10 nodes are shown in Fig. 5, for $d=10 \text{ and }\overline{M}=2$.

\begin{figure}
    \centering
    \includegraphics[width=0.6\linewidth]{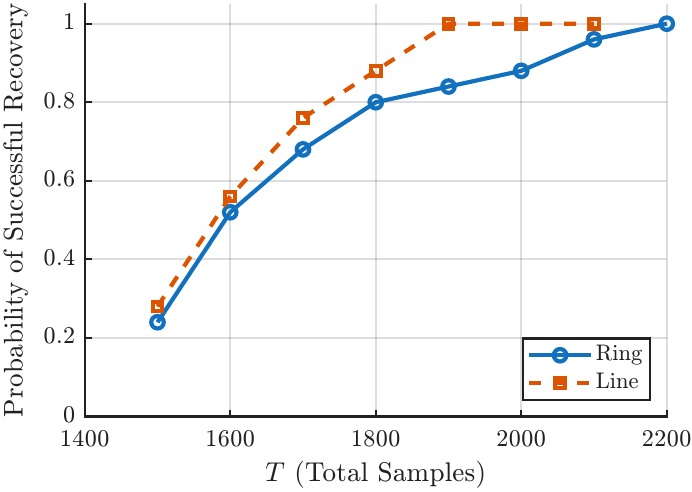}
    \caption{$|V| = 5, d = 5$, and $\overline{M} = 1$}
    \label{fig:1}
\end{figure}

\begin{figure}
    \centering
    \includegraphics[width=0.6\linewidth]{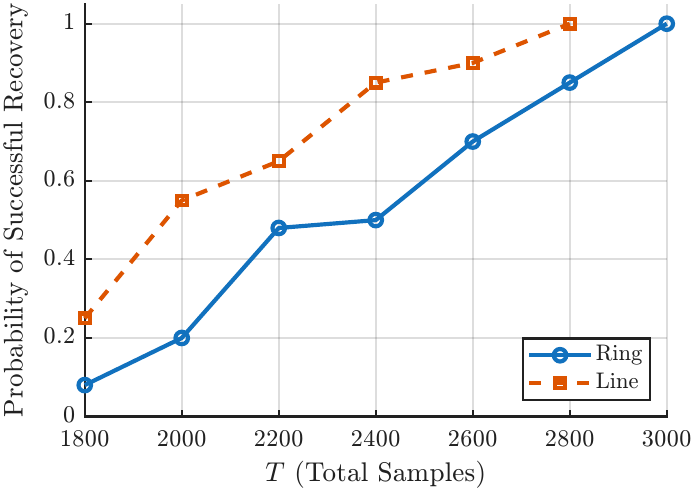}
    \caption{$|V| = 10, d = 10$, and $\overline{M} = 2$}
    \label{fig:2}
\end{figure}

\begin{figure}
    \centering
    \includegraphics[width=0.414\linewidth]{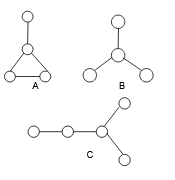}
    \caption{A) $K_3$ with a pendant, B) $S_3$, and C) Tree $|V| = 5$.}
    \label{fig:4}
\end{figure}

\begin{figure}
    \centering
    \includegraphics[width=0.6\linewidth]{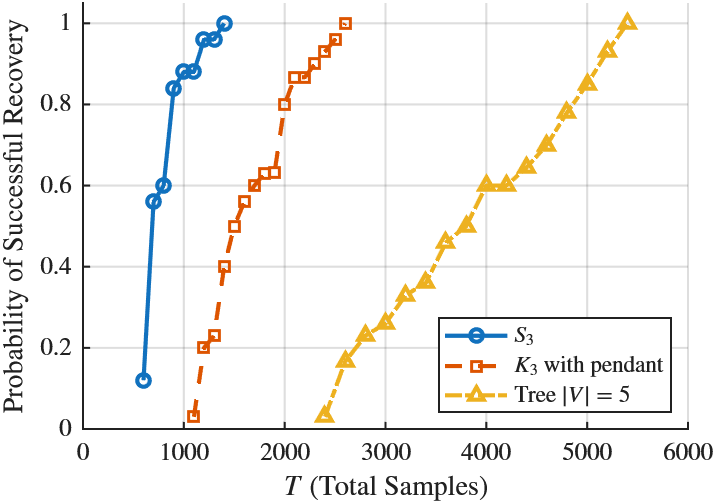}
    \caption{$d = 10$ and $\overline{M} = 1$}
    \label{fig:4}
\end{figure}

\begin{figure}
    \centering
    \includegraphics[width=0.6\linewidth]{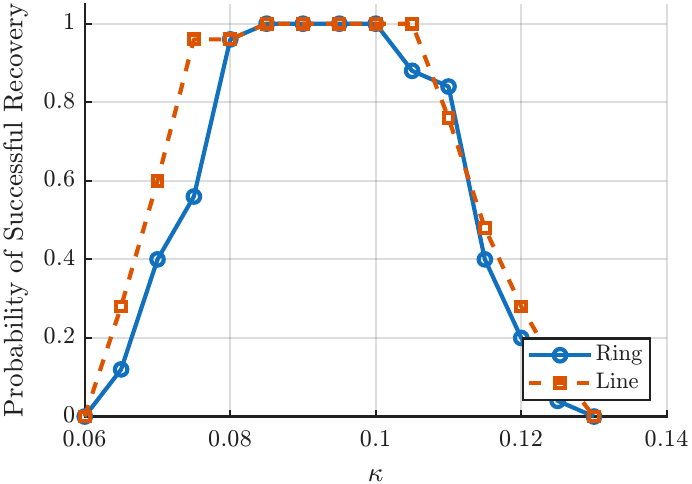}
    \caption{$|V| = 10, d = 10$, $T = 3500$, and $\overline{M} = 2$}
    \label{fig:4}
\end{figure}
 
\section{Conclusion}
In many applications, such as social networks and financial risk analysis, a distant past state can occasionally and randomly affect the future. Motivated by this,  we propose the past influence model (PIM)  by adapting a general Markovian model \cite{bagewadi2024learninginfluencegraphhighdimensional, 8345299} and propose an algorithm to learn the underlying influence graph of this process with a provable guarantee on accuracy and sample complexity. The algorithm has a certain threshold $\kappa$ (similar to regularization parameters), and we observe that, in experiments (Fig.~5), the algorithm's performance depends on this threshold. It was also observed that the influence graph
could be learned for constant values of $1 - p$ (Probability of
resetting), but its analysis is deferred for future work. This work is just the first step towards generalizing this model to capture the specific issues of social networks more effectively. It would be a challenging problem, but with direct applications.

\bibliographystyle{IEEEtran}        
\bibliography{main}

\appendix

For completeness, we will show how the data is processed under the knowledge of instants when resets occurred, which shall be used to obtain the probability estimates $p_g$ introduced in Section~\ref{sec: alg}. The sequence in (\ref{eq:PIM_pairs}) will be paired as following (\ref{eq:genie_pairs}) to obtain $p_g$ estimates.
\begin{align}
\label{eq:genie_pairs}
   &[y(1),y(2)],.....,\underset{\textbf{\textcolor{red}{Tail}}} {[y(1),y^{(2)}(2)]}, \underset{\textbf{\textcolor{red}{Tail}}}{[y(2),y^{(2)}(3)]},..\\ \nonumber
   &[y^{(2)}(d+1),y(d+2)], \underset{\textbf{\textcolor{red}{Tail}}}{[y^{(2)}(2),y^{(3)}(3)]},[y^{(3)}(3),y^{(3)}(4)],..,\\ \nonumber
   &\underset{\textbf{\textcolor{red}{Tail}}}{[y^{(2)}(d+1),y^{(3)}(d+2)]},.., \underset{\textbf{\textcolor{red}{Tail}}}{[y^{(3)}(3),y^{(4)}(4)]},....
\end{align}

The following Lemma follows directly from the Markovian properties of the post-processing distribution under knowledge of resets.
\begin{Lemma}
\label{lem:CondEntMarkov}
For any node $v$ and any node $u \not\in \mathcal{N}_v$, if $Q \subset V$ is such that $\mathcal{N}_v \subseteq Q$, then:
\[
H_g(v_+|v,Q,u) = H_g(v_+|v,Q).
\]
Moreover, under Assumption~1, for any $Q \subset V$, if $u \in \mathcal{N}_v$ and $u \not\in Q$, then, for some constant $\kappa > 0$, 
\[
H_g(v_+|v,Q) - H_g(v_+|v,Q,u) > \kappa,
\]
Where $H_g(.|.)$ are the entropy obtained from the $p_g$ distribution.
\end{Lemma}
 Intuitively, the above Lemma states that a node is in the neighborhood of $v$ if and only if adding it to the conditioning set $Q$ leads to a reduction in the directed conditional entropy $H_g(v_+|v,Q)$. The algorithms we present rely on this principle to identify neighborhoods. However, when $Q$ does not fully contain $\mathcal{N}_v$, Lemma~\ref{lem:CondEntMarkov} does not prevent the possibility of false positives—nodes being incorrectly identified as neighbors. The algorithm is designed to handle this issue by running nested loops leveraging Lemma~\ref{lem:tau_modified}.

Similar to $\phat, \hat{p}_g$ is the naive plugin estimate of the distribution $p_g$, and the entropy estimates $\hat{H}_g$ are obtained from $\hat{p}_g$. Drawing from the ideas in \cite{7097023}, \cite{bagewadi2024learninginfluencegraphhighdimensional} shows that if the estimation error in directed conditional entropy is kept sufficiently small, it is possible to select a threshold to reliably distinguish between neighbors and non-neighbors. This is formally stated in the following Lemma for our problem.

\begin{Lemma}
    \label{lem:tau_modified}
    Under Assumption~1, if for every node $v$ and any subset $P \subseteq V$, $|\Hhat(v_+|v,P) - \Hhat_g(v_+|v,P)| \leq \frac{\epsilon'}{4}, \text{  and  } |\Hhat_g(v_+|v,P) - H_g(v_+|v,P)| \leq \frac{\epsilon}{4}$. \\ 
    \\
    For all $u \in \mathcal{N}_v$, $v \in V$, and $Q \subset V \setminus\{u,v\}$,
    \begin{equation}
    \label{thresh1}
       \Hhat(v_+|v,Q) - \Hhat(v_+|v,Q,u) > \frac{\kappa}{2}
    \end{equation}
    
    and for all $w \notin \mathcal{N}_v$, with $Q \subset V \setminus\{u,v,w\}$ such that $\mathcal{N}_v \subseteq Q$,
    \begin{equation}
    \label{thresh2}
        \Hhat(v_+|v,Q) - \Hhat(v_+|v,Q,w) < \frac{\kappa}{2}
    \end{equation}
\end{Lemma}

\begin{proof}
    
Based on the conditions of Lemma~\ref{lem:tau_modified}, 
\[
\hat{H}_g(v_+|v,P) - \frac{\epsilon}{4} \leq \hat{H}(v_+|v,P) \leq \hat{H}_g(v_+|v,P) + \frac{\epsilon}{4}
\] 
and a similar entropy sandwich can be obtained for $\Hhat(v_+|v,P,u), \Hhat_g(v_+|v,P,u),$ and $\Hhat_g(v_+|v,P)$. Upon combining these inequalities: For all $u \in \mathcal{N}_v$, $v \in V$, and $Q \subset V$  $\Hhat(v_+|v,Q) - \Hhat(v_+|v,Q,u) \geq \kappa - \left( \frac{\epsilon}{2}+\frac{\epsilon'}{2} \right)$, upon having $\frac{\kappa}{2} = \left( \frac{\epsilon}{2}+\frac{\epsilon'}{2} \right)$, the Lemma is proved. A similar argument can be made for the second part as well.
\end{proof}

\vspace{3mm}

The proof of Theorem~1 has three main steps. First, Lemma~\ref{lem:7_infoth} shows that controlling the probability estimation error also controls the entropy estimation error—crucial for setting the threshold $\kappa$ to distinguish neighbors from non-neighbors, as in Lemma~\ref{lem:tau_modified}. Lemma~\ref{lem: w_bound} upper bounds the probability of estimation error between a Markov chain without any intermittent resets and its stationary distribution. Next, Lemma~\ref{lem:soltn}, the paper’s main technical result, bounds the probability estimation error for a Markov process that probabilistically forgets recent samples and specifies the minimum sample size needed to keep this error below a target level with probability at least $1 - \gamma$. Together, these results establish Theorem~1.

\begin{Lemma}
\label{lem:7_infoth}
The upper bound on the entropy estimation error
\[
\big| \hat{H}_g(v_+ \mid v, P) - H_g(v_+ \mid v, P) \big| < \frac{\epsilon}{4}
\]
holds if the following conditions are satisfied:
\[
\left\| \hat{p}_g(y_{v_+}, y_{v,P}) - p_g(y_{v_+}, y_{v,P}) \right\|_1 < \delta, 
\quad \forall v \in V, \; P \subset V,
\]
when $\delta \, \log\!\left(\tfrac{|\xi|}{\delta}\right) \leq \tfrac{\epsilon}{4}.$
Similarly, if $\delta' \, \log\!\left(\tfrac{|\chi|}{\delta'}\right) \leq \frac{\epsilon'}{4}$
\[
\left\| \hat{p}(y_{v_+}, y_{v,P}) - \hat{p}_g(y_{v_+}, y_{v,P}) \right\|_1 < \delta',
\quad \forall v \in V, \; P \subset V,
\]
then
\[
\big| \hat{H}(v_+ \mid v, P) - \hat{H}_g(v_+ \mid v, P) \big| < \epsilon'.
\]
\end{Lemma}

\begin{proof}
Applying theorem 17.3.3 from \cite{cover2006elements} for $\phat_g(y_{v_+},y_{v,P})$, and $ p_g(y_{v_+},y_{v,P})$, 
\begin{align}
\label{eq:infotheory}
   &|\Hhat_g(v_+,v,P)-H_g(v_+,v,P)| \leq  -\|\phat_g(y_{v_+},y_{v,P})- \\ \nonumber 
    & p_g(y_{v_+},y_{v,P})\|_1 \times \log\frac{\|\phat_g(y_{v_+},y_{v,P})-p_g(y_{v_+},y_{v,P})\|_1}{|\xi|}  
\end{align}

In Eqn. (\ref{eq:infotheory}) above, $|\xi|$ is used instead of $|\xi_{P}|$, which is the alphabet size of set $P$ as $|\xi| \geq |\xi_P|$. The R.H.S of Eqn.~(\ref{eq:infotheory}) is of the form: $f(t) = t.log(\frac{c}{t})$. Note that,$f(t)$ increases when $t<\frac{c}{e}$. Thus,
\begin{align}
   &|\Hhat_g(v_+,v,P)-H_g(v_+,v,P)| \leq \delta.log\left(\frac{|\xi|}{\delta} \right) \leq \frac{\epsilon}{4}
\end{align}
The same proof structure applies for the second part of the Lemma as well.    
\end{proof}

Let $\lambda^*$ be the operator norm of the Markov kernel of \{W(t)\} on zero-mean square-integrable functions (as in \cite{JMLR:v22:19-479}). \{W(t)\} is same as \textbf{PIM} except that there are no intermittent resets. Following Lemma is a modification of Lemma 8 in \cite{bagewadi2024learninginfluencegraphhighdimensional}, for $\{W(t)\}$. In the following Lemma, $w_x$ has the same interpretations as $y_x$, and $\pi$ is the stationary distribution of $\{W(t)\}$, and $\hat{p}_w$ is the empirical estimate obtained from the samples of $\{W(t)\}$.
\begin{Lemma} \label{lem: w_bound}
If $|\lambda^*| < 1$, following is true $\forall \, \Delta>0$:
\begin{align}
&\prob\!\left( \|\hat{p}_w(w_{v,+}, w_{v,P}) - \pi(w_{v,+}, w_{v,P})\|_1 > \Delta \right) \\
&\leq 2|V||\xi|\!\left(\frac{|V|}{|P|_{\max}}\right) \times \nonumber \\
   & \exp\!\left(\frac{-2(1-|\lambda^*|)(T-1)\Delta^2}{(1+|\lambda^*|)(|\xi|^2)}\right) \nonumber
\end{align}
\end{Lemma}

\begin{proof}
    Consider the Markov chain $Z(t) = (W(t), W(t-1))$. 
We define a function 
\begin{align*}
f(Z(t)) 
&= \frac{\mathbb{1}\big( (W(t), W_{v,P}(t-1)) = (w_{v,+}, w_{v,P}) \big)}{T-1} \nonumber \\
\E f(Z(t)) &= \frac{\pi(w_{v+}, w_{v,P})}{T-1}
\end{align*}
where $\mathbb{1}$ is the indicator function. Clearly, the range of f is $\left[0, \frac{1}{T-1}\right]$.Now, using the Hoeffding concentration inequality for Markov chains stated in Theorem 3 of \cite{JMLR:v22:19-479} on $f(Z(t)$, we get $\forall \Delta_1 >0$
\begin{align}
&\prob\!\left( \left|\frac{1}{T-1}\sum_{t=1}^{T} \left( f(Z(t)) - \mathbb{E}f(Z(t)\right))\right| > \Delta_1 \right) \\
&\leq 2 \exp\!\left( -\frac{2(1-|\lambda^*|)(T-1) \Delta_1^2}{(1+|\lambda^*|)} \right) \nonumber
\end{align}

By applying union bound on $(w_{v,+}, w_{v,p})$, we have $(\Delta = |\xi| \Delta_1)$
\begin{align}
&\prob\!\left( \|\hat{p}_w(w_{v,+}, w_{v,p}) - \pi(w_{v,+}, w_{v,p})\|_1 > \Delta \right) \\
&\leq 2|\xi| \exp\!\left( -\frac{(1-|\lambda^*|)2(T-1)(\Delta/|\xi|)^2}{(1+|\lambda^*|)} \right) \nonumber
\end{align}

which implies
\begin{align}
&\prob\!\left( \|\hat{p}(w_{v,+}, w_{v,p}) - \pi(w_{v,+}, w_{v,p})\|_1 > \Delta \right) \\ 
&\leq 2|\xi| \exp\!\left( -\frac{2(1-|\lambda^*|)(T-1)\Delta^2}{(1+|\lambda^*|)(|\xi|^2)} \right) \nonumber
\end{align}

Using the result in Lemma 6 of \cite{bagewadi2024learninginfluencegraphhighdimensional} that $\lambda^* \le 2(\bar{\mu} + L) |\rho(\tilde{A})|$, and applying union bound on $v$ and $P$, 
\begin{align}
&\prob\!\left( \|\hat{p}(w_{v,+}, w_{v,P}) - \pi(w_{v,+}, w_{v,P})\|_1 > \Delta \right) \\
&\leq 2|V|\!\left(\frac{|V|}{|P|_{\max}}\right)|\xi|
   \exp\!\left(-\frac{2(1-2(\bar{\mu} + L) |\rho(\tilde{A})|)(T-1)\Delta^2}{(1+2(\bar{\mu} + L) |\rho(\tilde{A})|)(|\xi|^2)}\right) \nonumber
\end{align}

\end{proof}
Following similar steps to those in Lemmas~9 and 10 of \cite{bagewadi2024learninginfluencegraphhighdimensional}, an upper bound on the distance from the stationary distribution is derived in terms of the parameters of the dynamics and is compared with a lower bound on the same in terms of $|\lambda^*|$, which is $|\lambda^*| \le 2(\bar{\mu} + L) |\rho(\tilde{A})|$ 
\begin{Lemma}
    \label{lem:soltn}
     If the conditions mentioned in Theorem. 1 are satisfied, and under Assumption 1,$\|\phat(y_{v_+},y_{v,P})-\phat_g(y_{v_+},y_{v,P})\|_1 < \delta'$  and $\|p_g(y_{v_+},y_{v,P})-\phat_g(y_{v_+},y_{v,P})\|_1 < \delta$ with probability at least $1-\gamma$,$\forall v \in V, P\subset V$, where $\pi(.)$ is the stationary distribution of $Y(t)$. 
\end{Lemma}

\begin{proof}
Clearly, using union bound, we need to upper-bound,  
\begin{align}    
    \label{eq:union_bound1_1}
    & \leq \prob(\|\phat(y_{v_+},y_{v,P})- \phat_g(y_{v_+},y_{v,P})\| > \delta')\\
    \label{eq:union_bound1_2}
    & + \prob(\|\phat_g(y_{v_+},y_{v,P})-p_g(y_{v_+},y_{v,P})\|_1>\delta)
\end{align} 

In (\ref{eq:union_bound1_1}), all the probability estimates that we require are of the form $p(y_{v_+}|y_{v,S})$ for some $S \subset V$: The samples have to be paired across time before estimation: $\{y_{v,S}(t), y_{v,S}(t+1)\}_{t=0}^{T-1}$ as in (\ref{eq:PIM_pairs}). Hence, there are $T-1$ total samples available for the estimation of $\hat{p}$ and $\hat{p}_g$. Thus $\| \hat{p} -\hat{p}_g \|$ will be a multiple of $\frac{2}{T-1}$. Since the discrepancy between $\hat{p}$ and $\hat{p}_g$ is due to  the recurring tails, (\ref{eq:union_bound1_1}) is:\\
\begin{align}
\nonumber
&= \prob(\|\phat(y_{v_+},y_{v,P})- \phat_g(y_{v_+},y_{v,P})\| > \delta') \\
&\leq \prob\left(\# tails \geq \frac{\delta'(T-1)}{2}\right)
\end{align}
Now, applying Chernoff bound for binomial distribution with $\mu = (1-p)(T-d-1)=\frac{(T-1)^{\alpha}}{\beta_1}$. Recall, that $\alpha<1$, and $0<\beta_1<1$.So, the above inequality is
\begin{align}
\nonumber
&\leq \exp\left(-\mu \left(\frac{\delta'(T-1)}{4\mu} -1\right)^2 \left(\frac{\delta'(T-1)}{4\mu} +2\right)^{-1} \right) \\
\label{eq: chernoff}
& = \exp\left(-\frac{(T-1)}{4}\frac{\left(\delta' - 4(T-1)^{\alpha-1}\right)^2} {\left(\delta' +8(T-1)^{\alpha-1}\right)} \right)
\end{align}

Using the condition that $\delta' - 4(T-1)^{\alpha-1}> c_1>0$, and for some small $c>0$, Inequality.~(\ref{eq: chernoff}) is
\begin{align}
&\leq \exp\left(-\frac{(T-1)}{4}\frac{\left( c_1\right)^2} {\left(\delta' +  2\delta' - 2c_1\right)} \right) \\
    \label{eq: first}
&\leq c\exp{\left(-\frac{Tc_1^2}{12\delta'\beta_1}\right)}
\end{align}

Notice that the (\ref{eq:union_bound1_2}) resembles the probability of estimation error between the empirical distribution $\hat{p}_g$, and its stationary distribution. So, to bound the second term we look at the worst case scenario—when all tosses resulted in tails—upon forming pairs by skipping $d$ time steps from the post-processing, effectively reduces its run to a Markov chain resembling the standard process $\{W(t)\}$(check Example. 1). The resulting sequence has an effective length given by
\[
T_{\min} = \frac{T + d^2 - 1}{d + 1}.
\]

Let $\hat{p}_{rM}(\cdot)$ denote the probability estimates obtained using only the collated $T_{\min}$ samples. The stationary distribution of $Y(t)$, denoted $\pi(\cdot)$, is identical to that of the Markov chain $\{W(t)\}$. Moreover, $\hat{p}_{rM}(\cdot)$ matches exactly the estimates from the same run of $\{W(t)\}$, i.e., $\hat{p}_{rM} \equiv \hat{p}_w$, where $\hat{p}_w(\cdot)$ denotes estimates from $\{W(t)\}$.So,

\begin{align}
   & \ \ \ \ \ \,\prob(\|\phat_g(y_{v_+},y_{v,P})-p_g(y_{v_+},y_{v,P})\|_1>\delta)\\
   & \leq \prob\left(\|\phat_{rM} (y_{v_+},y_{v,P})-p_g(y_{v_+},y_{v,P})\|_1>\delta\right) \\
   \nonumber
   & = \prob\left(\|\phat_{w} (y_{v_+},y_{v,P})-\pi(y_{v_+},y_{v,P})\|_1>\delta\right)
\end{align}
Thus, above Inequality equals the deviation between the empirical and stationary distributions of $\{W(t)\}$.So, Inequality~(\ref{eq: smita_eqn}) follows directly from Lemma~\ref{lem: w_bound}:
\begin{align}
    \label{eq: smita_eqn}
    & \leq 2\cdot |V|^{(|P|_{\max}+1)} \times |\xi| \exp\left( -\frac{2\left(1 - |\lambda^{*}|\right)(T_{min}-1)\delta^2}{\left(1 + |\lambda^{*}|\right)|\xi|^2} \right) 
\end{align}

Combining the bounds in~(\ref{eq: first}) and~(\ref{eq: smita_eqn}) yields the following upper bound on (\ref{eq:union_bound1_1}) + (\ref{eq:union_bound1_2}):
\begin{align}
    \label{eq: final_probab_bound}
    &\leq \left( c + 2. |V|^{(P_{max} +1)} \times |\xi|\right) \times \\ \nonumber & \max{\left\{\exp\left( -\frac{\left(1 - |\lambda^{*}|\right)(T_{min}-1)\delta^2}{2\left(1 + |\lambda^{*}|\right)|\xi|^2} \right),\exp{\left(-\frac{Tc_1^2}{12\delta'\beta_1}\right)} \right\}}
\end{align}

If T is large enough such that Inequality.~(\ref{eq: final_probab_bound}) is at most $\gamma$ (Theorem. 1), $\|\phat(y_{v_+},y_{v,P})-\phat_g(y_{v_+},y_{v,P})\|_1 < \delta'$  and $\|p_g(y_{v_+},y_{v,P})-\phat_g(y_{v_+},y_{v,P})\|_1 < \delta$, $\forall v \in V, P\subset V$ with probability at least $1-\gamma$, thus, satisfying Lemma~\ref{lem:7_infoth},which in turn provides the threshold to the PIMRecGreedy algorithm to distinguish between the neighbors and non-neighbors, according to Lemma~\ref{lem:tau_modified}.
\end{proof}

\Example{A worst-case scenario} \\ 
When $d=2$, and all coin tosses were tails, the data would be paired up as the following sequence under the knowledge of instants when resets occurred (in this case, everywhere).

\begin{align} \nonumber
    &[y(1),y(2)],[y(2),y(3)],[y(1),y^{(2)}(2)],[y(2),y^{(2)}(3)], \\ \nonumber &[y(3),y(4)], [y^{(2)}(2),y^{(3)}(3)], [y^{(2)}(3),y^{(2)}(4)],\\ \nonumber
     &[y(4),y(5)],[y^{(3)}(3),y^{(3)}(4)],... \\ \nonumber
\end{align}

On picking up a pair after every $d=2$ pairs, will result in the following sequence, 
 \begin{align} \nonumber
    [y(1),y^{(2)}(2)],[y^{(2)}(2),y^{(3)}(3)],[y^{(3)}(3),y^{(3)}(4)],..
 \end{align}
Note that the above sequence resembles a normal Markov run just like that of the linear Markov model ${W(t)}$.

\end{document}